\documentclass[11pt]{article}

\usepackage{amsmath,amsfonts,amsthm,amssymb,bm,tikz-cd}
\usepackage[mathscr]{eucal}
\usepackage[margin=1.1in]{geometry}
\usepackage[english]{babel}
\usepackage{multirow}
\usepackage{authblk}
\usepackage{tikz-cd} 
\usepackage{url}
\usepackage{color}
\usepackage{hyperref}

\usepackage{enumitem}\setlist[enumerate]{label=\arabic*.} 

\theoremstyle{plain}

\newtheorem{corollary}{Corollary}
\newtheorem{lemma}{Lemma}
\newtheorem{proposition}{Proposition}
\newtheorem{problem}{Problem}
\newtheorem{theorem}{Theorem}
\newtheorem{claim}{Claim}
\newtheorem*{acknowledgments}{Acknowledgments}
\theoremstyle{definition}
\newtheorem{definition}{Definition}
\newtheorem{example}{Example}

\newcommand{\ord}[1]{\textup{ord}\left(#1\right)}
\newcommand{\supp}[2]{\textup{supp}_{#1}\left(#2\right)}

\begin{document}

\title{On the order of lazy cellular automata}
\author[1]{Edgar Alcal\'a-Arroyo\footnote{Email: edgar.alcala7434@alumnos.udg.mx }}
\author[1]{Alonso Castillo-Ramirez\footnote{Email: alonso.castillor@academicos.udg.mx }}
\affil[1]{Centro Universitario de Ciencias Exactas e Ingenier\'ias, Universidad de Guadalajara, M\'exico.}

\maketitle

\begin{abstract}
We study the most elementary family of cellular automata defined over an arbitrary group universe $G$ and an alphabet $A$: the \emph{lazy cellular automata}, which act as the identity on configurations in $A^G$, except when they read a \emph{unique active transition} $p \in A^S$, in which case they write a fixed symbol $a \in A$. As expected, the dynamical behavior of lazy cellular automata is relatively simple, yet subtle questions arise since they completely depend on the choice of $p$ and $a$. In this paper, we investigate the \emph{order} of a lazy cellular automaton $\tau : A^G \to A^G$, defined as the cardinality of the set $\{ \tau^k : k \in \mathbb{N} \}$. In particular, we establish a general upper bound for the order of $\tau$ in terms of the fibers of the pattern $p$, and we prove that this bound is attained when $p$ is a quasi-constant pattern.    \\

\textbf{keywords}: Lazy cellular automaton; unique active transition; order of a cellular automaton; quasi-constant pattern.
\end{abstract}


\section{Introduction}

A cellular automaton (CA) is a mathematical model over a discrete space defined by a local map that is applied homogeneously and simultaneously to the whole space. The underlying discrete space is a \emph{configuration space} $A^G$, which consists of all maps from a group universe $G$ to an alphabet $A$. Following \cite{CSC10}, a cellular automaton $\tau : A^G \to A^G$ is a function such that there exist a finite subset $S \subseteq G$, called a \emph{neighborhood} of $\tau$, and a \emph{local map} $\mu :A^S \to A$ satisfying
\[ \tau(x)(g) = \mu( (g\cdot x) \vert_S), \quad \forall \ x \in A^G, g \in G, \]    
where $\cdot$ denotes the \emph{shift action} of $G$ on $A^G$:
\[  (g \cdot x)(h):= x(hg), \quad \forall h \in G. \]

We say that a cellular automaton $\tau : A^G \to A^G$ is \emph{lazy} if there is a local defining map $\mu : A^S \to A$ for $\tau$ such that $e \in S$, where $e$ is the identity of the group $G$, and there exists a pattern $p \in A^S$, called the \emph{unique active transition} of $\tau$, satisfying the following: 
\[ \forall z \in A^S, \quad \mu(z) = z(e) \Longleftrightarrow z \neq p. \]
Intuitively, a lazy cellular automaton acts almost as the identity function of $A^G$, except when it reads the pattern $p$, in which case it writes the symbol $a:=\mu(p) \in A \setminus \{ p(e) \}$.  

As expected, the dynamical behavior of a lazy CA is relatively simple, yet subtle questions arise because their evolution depends entirely on the choice of $p$ and $a$. In a certain sense, lazy cellular automata are even more elementary than the well-known \emph{elementary cellular automata} (ECA) studied by Wolfram \cite{Wolfram}, where complex behavior already emerges. Since there are only 256 ECA, a complete case-by-case computational analysis is feasible, whereas there are infinitely many lazy CAs (over an infinite group universe), as their neighborhood size can be arbitrarily large. Despite this, a deep understanding of lazy CAs is possible and may lead to new insights into broader families of cellular automata. 

Lazy cellular automata were introduced in \cite{Idem} as a tool to study \emph{idempotent} CAs: this is, cellular automata $\tau : A^G \to A^G$ satisfying $\tau^2 = \tau$. It was observed that if the unique active transition $p \in A^S$ is constant (i.e., $p(e)=p(s)$, $\forall s \in S$) or symmetric (i.e., $S=S^{-1}$ and $p(s)=p(s^{-1})$, $\forall s \in S$), then $\tau$ is idempotent. Moreover, the idempotence of $\tau$ was completely characterized when $p$ is \emph{quasi-constant}, meaning that there is $r \in S$ such that $p \vert_{S \setminus \{r\}}$ is constant.    

Here we study the \emph{order} of a lazy cellular automaton $\tau : A^G \to A^G$, denoted by $\ord{\tau}$, as the cardinality of the set of all powers of $\tau$:
\[ \ord{\tau} := \vert \{ \tau^k : k \in \mathbb{N} \} \vert. \]
 Besides being an important algebraic concept, the order also captures part of the dynamical behavior of a cellular automaton. For example, it was shown by K\r{u}rka \cite[Theorem 4]{Kurka} that a \emph{one-dimensional} cellular automaton (i.e, when the underlying universe is the group of integers $\mathbb{Z}$) is of finite order if and only if it is equicontinuous.

The dynamical behavior of one-dimensional lazy CAs was examined in \cite{One-dim}. In this case, when the neighborhood $S \subseteq \mathbb{Z}$ of $\tau$ is an interval, an interesting dichotomy arises: either $\tau$ is idempotent or of infinite order, which is equivalent to being strictly almost equicontinuous as a dynamical system.    
 
In this paper, we establish several results on the order of lazy CAs in the general setting of an arbitrary group universe $G$. In Section 2, we show that for any lazy cellular automaton $\tau : A^G \to A^G$, if $\ord{\tau}$ is finite, then $\tau$ has period $1$, which means that there is $n \in \mathbb{N}$ such that $\tau^n = \tau^{n+1}$. Moreover, in Theorem \ref{upperbound}, we derive a general upper bound for $\ord{\tau}$ which depends purely on combinatorial properties of the fibers $S_a := p^{-1}\{a\} \subseteq S$ of of its unique active transition $p \in A^S$. As a consequence, in Corollary \ref{cor-idem}, we provide a sufficient condition on $p \in A^S$ that guarantees that $\tau$ is idempotent. 

In Section 3, we determine the order of a lazy CA whose unique active transition is a quasi-constant pattern, which is a substantial generalization of Theorem 2 in \cite{Idem}. 

As an easy consequence of Theorem \ref{th-main} we find that, in contrast with the dichotomy obtained in \cite{One-dim} when the neighborhood $S \subseteq \mathbb{Z}$ is an interval, for every $n \geq 2$, there exists a lazy cellular automaton $\tau : A^\mathbb{Z} \to A^\mathbb{Z}$ such that $\ord{\tau} =n$. 

Finally, in Section 3, we present two open problems related to the study of lazy cellular automata.


\section{The order of lazy CA}

We assume that the alphabet $A$ has at least two different elements $0$ and $1$. A \emph{pattern} is a function $p \in A^S$, where $S$ is a finite subset of the group universe $G$. We denote patterns by $p=p(s_1)p(s_2)\dots p(s_n)$, where we fix an order on the set $S = \{ s_1, s_2, \dots, s_n\}$. 

The term \emph{active transitions} of a local map $\mu : A^S \to A$, with $e \in S$, has been recently introduced by several authors \cite{Pedro1, Pedro2, Fates1, Fates2} as the patterns $z \in A^S$ such that $\mu(z) \neq z(e)$. The \emph{activity value} \cite{Concha} of $\mu : A^S \to A$ is simply the number of active transitions of $\mu$. 

\begin{definition}\label{lazy}
A cellular automaton $\tau: A^G \to A^G$ is called \emph{lazy} if there is a local defining map $\mu : A^S \to A$ for $\tau$ such that $e \in S$ and there exists $p \in A^S$ satisfying 
\[ \forall z \in A^S, \quad \mu(z) = z(e) \Longleftrightarrow z \neq p. \]
In such case, we say that $p$ is the \emph{unique active transition} of $\tau$ and $a:=\mu(p) \in A \setminus \{p(e)\}$ is the \emph{writing symbol} of $\tau$.
\end{definition}

The adjective ``lazy'' is justified for this class of cellular automata as they have the smallest non-zero activity value (i.e., activity value $1$). 

\begin{example}\label{ex-236}
Let $G=\mathbb{Z}$ and $A=\{0,1\}$. The ECA rule 236 (see \cite[Sec. 2.5]{Kari} for an explanation of the rule labeling) is lazy, because it may be defined via the local map $\mu : A^{\{-1,0,1 \}} \to A$ given by the following table:
\[\begin{array}{c|cccccccc}
z \in A^{\{-1,0,1\}} & 111 & 110 & 101 & 100 & 011 & 010 & 001 & 000 \\ \hline
\mu(z) \in A & 1 & 1 & 1 & 0 & 1 & 1 & 0 & 0
\end{array}\]
The unique active transition of rule 236 is $p=101 \in A^{\{-1,0,1\}}$. 
\end{example}

\begin{example}
Let $G=\mathbb{Z}$ and $A=\{0,1\}$. The ECA rule 136 may be defined via the local map $\mu : A^{\{-1,0,1 \}} \to A$ given by the following table
\[\begin{array}{c|cccccccc}
z \in A^{\{-1,0,1\}} & 111 & 110 & 101 & 100 & 011 & 010 & 001 & 000 \\ \hline
\mu(z) \in A & 1 & 0 & 0 & 0 & 1 & 0 & 0 & 0
\end{array}\]
It may seem that rule 136 is not lazy because the above local map has two active transitions $110$ and $010$. However, this local map may be reduced to its minimal neighborhood, so we obtain a local defining map $\mu^\prime : A^{\{0,1\}} \to A$ for rule 136 given by the following table:
\[\begin{array}{c|cccccccc}
z \in A^{\{0,1\}} & 11 & 10 & 01 & 00  \\ \hline
\mu^\prime(z) \in A & 1 & 0 & 0 & 0 
\end{array}\]
 Therefore, rule 136 is lazy and has unique active transition $p = 10 \in A^{\{0,1\}}$. 
\end{example}

Recall that the \emph{minimal neighborhood} of a cellular automaton $\tau : A^G \to A^G$ is a neighborhood admitted by $\tau$ of smallest cardinality, which always exists and is unique by \cite[Proposition 1.5.2]{CSC10}. Equivalently, the minimal neighborhood of $\tau$ is equal to the set of all elements of $G$ that are essential in order to define a local map for $\tau$ (see \cite[Proposition 1]{CV24}). We call the \emph{minimal local map} of $\tau$ to the local defining map of $\tau$ associated to its minimal neighborhood.    

It was shown in \cite[Lemma 1]{Idem} that if $\mu : A^S \to A$ has a unique active transition and $\vert S \vert \geq 2$, then the minimal neighborhood of the lazy cellular automaton $\tau : A^G \to A^G$ defined by $\mu$ is precisely $S$. Hence, the next result follows. 

\begin{proposition}
Let $\tau : A^G \to A^G$ be a non-constant cellular automaton with minimal local map $\mu : A^S \to A$. Then, $\tau$ is lazy if and only if $e \in S$ and $\mu$ has a unique active transition. 
\end{proposition}

We say that a pattern $p \in A^S$ \emph{appears} in $x \in A^G$ if there is $g \in G$ such that $(g \cdot x)\vert_S = p$. 

\begin{lemma}\label{le-equality}
Let $\tau: A^G\to A^G$ be a lazy CA with unique active transition $p \in A^S$. Then, $p$ appears in a configuration $x \in A^G$ if and only if $x \neq \tau(x)$.
\end{lemma}
\begin{proof}
 By Definition \ref{lazy}, there is $g \in G$ such that $(g \cdot x)\vert_S = p$ if and only if 
\[ \tau(x)(g) = \mu((g \cdot x)\vert_S)  \neq p(e) = (g \cdot x)(e) = x(g),  \]
where $\mu : A^S \to A$ is the minimal local map of $\tau$. The result follows. 
\end{proof}

For any cellular automaton $\tau : A^G \to A^G$ and $n \in \mathbb{N}$, we denote by $\tau^n$ the $n$-th composition of $\tau$ with itself: 
\[ \tau^n = \underbrace{\tau \circ \dots \circ\tau}_{n \text{ times }}. \]
By \cite[Prop. 1.4.9]{CSC10}, $\tau^n : A^G \to A^G$ is also a cellular automaton.

\begin{corollary}\label{cor-power}
Let $\tau: A^G\to A^G$ be a lazy CA with unique active transition $p \in A^S$. For any $n \in \mathbb{N}$, $\tau^{n} \neq \tau^{n+1}$ if and only if there exists $x \in A^G$ such that $p$ appears in $\tau^{n}(x)$.  
\end{corollary}

For a pattern $p \in A^S$, we denote the \emph{fiber} of $a \in A$ under $p$ by
 \[S_a := p^{-1}\{a\} =  \{ s\in S : p(s) = a\}. \]
 The next result follows by \cite[Lemma 3]{Idem} or \cite[Lemma 3.2]{One-dim}, but we add its proof here for completeness.

\begin{lemma}\label{le-ale}
Let $\tau: A^G\to A^G$ be a lazy CA with unique active transition $p \in A^S$ and writing symbol $a \in A$. If there exists $x \in A^G$ such that $\tau(x) \vert_S = p$, then $(s \cdot x)\vert_S = p$ for some $s \in S_a$.
\end{lemma}
\begin{proof}
The hypothesis $\tau(x) \vert_S = p$ is equivalent to 
\begin{equation}\label{hip-ale}
p(s) = \tau(x)(s) = \mu((s \cdot x) \vert_S), \quad \forall s \in S,  
\end{equation}
where $\mu : A^S \to A$ is the minimal local map of $\tau$. 

If $x \vert_S = p$, then $p(e) = \mu(x\vert_S) = a$ is a contradiction with Definition \ref{lazy}. Hence, $x \vert_S \neq p$ and $p(e) = \mu(x\vert_S) = x(e)$ by the hypothesis (\ref{hip-ale}) with $s=e$. Suppose that $(s \cdot x) \vert_S \neq p$ for all $s \in S \setminus \{e\}$. Then, by hypothesis (\ref{hip-ale}),
\[ p(s) = \mu((s \cdot x) \vert_S) = (s \cdot x)(e) = x(s), \quad \forall s \in  S \setminus \{e\}, \]
which contradicts that $x \vert_S \neq p$. Therefore, there must exist $s \in S \setminus \{e\}$ such that $(s \cdot x)\vert_S = p$. Again hypothesis (\ref{hip-ale}) gives us that that $p(s) = \mu((s \cdot x) \vert_S) = a$, so $s \in S_a$. 
\end{proof}

Given $x\in A^G$ and $b\in A$, define
\[\supp{b}{x} := \{g\in G : x(g) = b\}.\]
The following result presents some elementary properties of supports of configurations under lazy CAs. 

\begin{lemma}\label{supports}
    Let $\tau: A^G\to A^G$ be a lazy CA with unique active transition $p \in A^S$ and writing symbol $a \in A$. Let $i,j \in\mathbb{N}, i\leq j$. Then:
   \begin{enumerate} 
    \item $\supp{a}{\tau^i(x)}\subseteq \supp{a}{\tau^{j}(x)}$ for all $x \in A^G$. \label{L3-P1}
    \item  $\supp{p(e)}{\tau^{i}(x)}\supseteq \supp{p(e)}{\tau^{j}(x)}$ for all $x \in A^G$. Furthermore, 
    \[\supp{b}{\tau^{i}(x)}= \supp{b}{\tau^{j}(x)},\]
    for all $x \in A^G$, and $b\in A\setminus\{a, p(e)\}$. \label{L3-P2}
    \item If $ \supp{a}{x} =  \supp{a}{\tau(x)}$ for some $x \in A^G$, then $x= \tau(x)$. \label{L3-P3}
\end{enumerate}    
  \end{lemma}
\begin{proof}
   For parts (1.) and (2.), we shall prove the base case $i = 0$ and $j = 1$, as the rest follows by induction.  Let $g \in\supp{a}{x}$. Observe that $(g \cdot x)\vert_S \neq p$, because $(g \cdot x)(e) = x(g) = a \neq p(e)$. Definition \ref{lazy} implies
    \[ \tau(x)(g) = \mu( ( g \cdot x) \vert_S) = ( g \cdot x)(e) = x(g) = a,  \]
    where $\mu : A^S \to A$ is the minimal local map of $\tau$. This shows that $g \in \supp{a}{\tau(x)}$. Now, assume $b\in A\setminus\{a\}$ and $g\in \supp{b}{\tau(x)}$. Since $b\neq a$, Definition \ref{lazy} implies $\tau(x)(g) = x(g)$, so $g\in \supp{b}{x}$. Similarly, if $b\in A\setminus\{a, p(e)\}$ and $g\in\supp{b}{x}$, then $(g\cdot x)\vert_S \neq p$. As a consequence, $\tau(x)(g) = x(g)$ and $g\in\supp{b}{\tau(x)}$.

    For part (3), first note that $x(g) = a = \tau(x)(g)$ for all $g \in \supp{a}{x} = \supp{a}{\tau(x)}$. Let $h \in G \setminus \supp{a}{\tau(x)}$. This means that $\mu( (h \cdot x) \vert_S) = \tau(x)(h) \neq a$, so by Definition \ref{lazy}, $(h \cdot x) \vert_S \neq p$ and 
    \[ \tau(x)(h)  = \mu( (h \cdot x) \vert_S) = (h \cdot x)(e) = x(h). \] 
    Therefore, $x= \tau(x)$.
\end{proof}

If $\ord{\tau} < \infty$, let $m\in \mathbb{N}$ and $n\in\mathbb{Z}_+$ be as small as possible such that $\tau^m = \tau^{m + n}$. Then, $m$ and $n$ are referred to as the \emph{index} and \emph{period} of $\tau$, respectively. Note that $\ord{\tau} = m + n$, as depicted in the following diagram: 

\[\begin{tikzcd}[column sep=large]
\mathrm{id} \arrow[r] 
  & \tau \arrow[r] 
   & \cdots \arrow[r] 
  & \tau^m \arrow[r] 
  & \cdots \arrow[r] 
  & \tau^{m+n-1}
  \arrow[ll, bend right=35]
\end{tikzcd} \]

\begin{proposition}\label{period}
    Let $\tau: A^G\to A^G$ be a lazy CA. If $\ord{\tau} < \infty$, then the period of $\tau$ is 1.
\end{proposition}
\begin{proof} 
Let $\mu : A^S \to A$ be the minimal local map of $\tau$ with unique active transition $p \in A^S$ and writing symbol $a:=\mu(p) \in A$. If $\ord{\tau} < \infty$, let $m$ and $n$ be the index and period of $\tau$, so $\tau^{m} = \tau^{m+n}$. Lemma \ref{supports} (1.) implies that for all $x \in A^G$,
\[ \supp{a}{\tau^m(x)} \subseteq \supp{a}{\tau^{m+1}(x)} \subseteq \dots \subseteq \supp{a}{\tau^{m+n}(x)}. \]
However, since $\tau^{m} = \tau^{m+n}$, then $\supp{a}{\tau^m(x)} = \supp{a}{\tau^{m+n}(x)}$, so 
 \[ \supp{a}{\tau^m(x)} = \supp{a}{\tau^{m+1}(x)}, \quad \forall x \in A^G. \]
Hence, Lemma \ref{supports} (3.) implies that $\tau^m = \tau^{m+1}$, so, by the minimality of $n$, we must have $n = 1$.
\end{proof}

\begin{corollary}\label{order-lowerbound}
    Let $\tau: A^G\to A^G$ be a lazy CA and $n\in\mathbb{Z}_+$. Then,
   \[\ord{\tau} = \min\{n\geq 2: \tau^{n-1}=\tau^n\},\]
   where the minimum exists if and only if $\ord{\tau} < \infty$. 
\end{corollary}

We define the product of two subsets $S, K \subseteq G$ as $SK  := \{ s k : s \in S, k \in K \}$ and the inverse as $S^{-1} := \{ s^{-1} : s \in S \}$. 

 \begin{lemma}\label{le-fiber}
 For all $g \in G$, $x \in A^G$, $p \in A^S$, it holds that
        \[(g\cdot x)\vert_S = p \quad \iff \quad S_b g\subseteq \supp{b}{x}, \ \forall b \in A. \]
 \end{lemma} 
 \begin{proof}
 Suppose that $(g\cdot x)\vert_S = p$, which is equivalent to $(x)(sg) = p(s)$, $\forall s \in S$. For any $b \in A$, if $s \in S_b$, then $x(sg) = p(s) = b$. This shows that $S_b g\subseteq \supp{b}{x}$ for all $b \in A$.
 
 Conversely, suppose that $S_b g\subseteq \supp{b}{x}$, for all $b \in A$. Then, for all $s \in S$,
 \[ sg \in S_{p(s)}g \subseteq \supp{p(s)}{x}.  \] 
 This means that $x(sg) = p(s)$, for all $s \in S$, and the result follows.  
 \end{proof}

In order to state the main theorem of this section, we introduce the following notation about words on groups. A \emph{word} on $S \subseteq G$ of length $n$ is simply an element of the Cartesian power $S^{n} := \{ (s_1, \dots, s_n) : s_i \in S \}$, for $n \geq 1$ (we do not consider empty words). We consider the evaluation function $\theta$ from words on $S \subseteq G$ to elements of $G$ given by $\theta(s_1, \dots, s_n) := s_1 \dots s_n$. We say that $v$ is a \emph{subword} of $w = (s_1, \dots, s_n) \in S^n$, denoted by $v\sqsubseteq w$, if $v = (s_i, s_{i+1}, \dots, s_{j})$ for some $i \leq j$. For integers $i,j \in \mathbb{Z}$ such that $i \leq j$, we consider the integer interval $[i,j] := \{ k \in \mathbb{Z} : i \leq k \leq j \}$.  

 In the following result we provide a general upper bound for the order of a lazy cellular automaton in terms of a combinatorial condition satisfied by the fibers of its unique active transition. 

\begin{theorem}\label{upperbound}
    Let $\tau: A^G\to A^G$ be a lazy CA with minimal neighborhood $S \subseteq G$, unique active transition $p \in A^S$ and writing symbol $a \in A \setminus \{p(e)\}$. Then, $\ord{\tau}$ is at most the minimum $n\geq 2$ such that every word $w \in (S_a^{-1})^{n-1}$ has a subword $v\sqsubseteq w$ such that:
    \begin{enumerate}
        \item $\theta(v) \in S_b^{-1}S_a$, for some $b\in A\setminus\{a\}$; or \label{C1_upperbound}
        \item $\theta(v) \in S_{b_1}^{-1}S_{b_2}$, for some $b_1, b_2\in A\setminus\{a\}$, with $b_1\neq b_2$. \label{C2_upperbound}
    \end{enumerate}
\end{theorem}
\begin{proof}
We will show that for any $n\geq 2$, if $\tau^{n-1}\neq \tau^n$, then there exists a word $w=(s^{-1}_1, \ldots, s^{-1}_{n-1})\in (S^{-1}_a)^{n-1}$ such that (\ref{C1_upperbound}) and (\ref{C2_upperbound}) do not hold for all subwords $v \sqsubseteq w$, so the result follows from Corollary \ref{order-lowerbound}.

If $\tau^{n-1}\neq \tau^n$, then, by Corollary \ref{cor-power}, there exists $x \in A^G$ such that $p$ appears in $\tau^{n-1}(x)$. By the $G$-equivariance of cellular automata (i.e., $\tau(g \cdot x) = g \cdot \tau(x)$, for all $g \in G$, see \cite[Prop. 1.4.4]{CSC10}), we may assume that $\tau^{n-1}(x)\vert_S = p$. Applying iteratively Lemmas \ref{le-equality} and \ref{le-ale} yields the existence of $s_1, \dots, s_{n-1} \in S_a$ such that
        \[\big((s_j\cdots s_1)\cdot \tau^{(n-1)-j}(x)\big)\vert_S = p, \quad \forall j\in [1,  n-1] .\]
      By Lemma \ref{le-fiber} applied to the previous equality we obtain that for all $b\in A$, 
            \begin{align*}
            S_b(s_j\cdots s_1) &\subseteq \supp{b}{\tau^{(n-1)-j}(x)}, \quad \forall j\in [1,  n-1],
        \end{align*}
Now, Lemmas \ref{supports} (\ref{L3-P1}) and (\ref{L3-P2}) imply that for all $k \in  [1,  n-1]$,
        \begin{align*}
           \bigcup_{j=1}^k\big(S_as_{(n-1)-j}\cdots s_1\big)&\subseteq\supp{a}{\tau^k(x)}, \\
            \bigcup_{j=k}^{n-1}\big(S_bs_{(n-1)-j}\cdots s_1\big)&\subseteq\supp{b}{\tau^k(x)}, \quad\forall b\in A\setminus\{a\}.
        \end{align*}
Since supports with respect to different elements of $A$ are disjoint, we obtain that
        \begin{align}
            \left[\bigcup_{i=1}^k\big(S_as_{(n-1)-i}\cdots s_1\big)\right]&\cap\left[\bigcup_{j=k}^{n-1}\big(S_bs_{(n-1)-j}\cdots s_1\big)\right] = \emptyset, \quad \forall b\in A\setminus\{a\}, \label{eq:C1_upperbound}\\
            \left[\bigcup_{i=k}^{n-1}\big(S_{b_1}s_{(n-1)-i}\cdots s_1\big)\right]&\cap\left[\bigcup_{j=k}^{n-1}\big(S_{b_2}s_{(n-1)-j}\cdots s_1\big)\right] = \emptyset, \quad \forall b_1, b_2\in A\setminus\{a\}, b_1\neq b_2. \label{eq:C2_upperbound}
        \end{align}
   Observe that condition \eqref{eq:C1_upperbound} is equivalent to $(s_j\cdots s_i)^{-1} = s_i^{-1} \dots s_j^{-1} \notin S_b^{-1}S_a$ for all $1 \leq i \leq j \leq n-1$ and $b \in A \setminus \{a\}$, while condition \eqref{eq:C2_upperbound} is equivalent to $(s_j\cdots s_i)^{-1} =s_i^{-1} \dots s_j^{-1} \notin S_{b_1}^{-1}S_{b_2}$, for all $1\leq i\leq j \leq n-1$, and all $b_1, b_2\in A\setminus\{a\}$, $b_1 \neq b_2$. This means that no subword of $w:= (s_1^{-1}, \dots, s_{n-1}^{-1}) \in (S_a^{-1})^{n-1}$ is neither in $S_b^{-1}S_a$ nor in $S_{b_1}^{-1}S_{b_2}$, and the result follows.
  \end{proof}

\begin{example}
Consider the the lazy CA $\tau : A^\mathbb{Z} \to A^\mathbb{Z}$ with minimal neighborhood $S = \{-1,0,1\}$, unique active transition $p = 101 \in A^S$ and writing symbol $a=1 \in A$ (this is ECA rule 236 used in Example \ref{ex-236}). As the group $\mathbb{Z}$ uses additive notation, then condition (1.) of Theorem \ref{upperbound} involves the set
\[ -S_0 + S_1 = \{0\} + \{-1,1\} = \{-1,1\}.  \] 
Since the evaluation of any word of length $1$ over $- S_1 = \{1,-1\}$ is in $-S_0 + S_1$, Theorem \ref{upperbound} shows that the order of $\tau$ is at most order $2$. Since $\tau$ may not have order $1$ because it is not the identity function, then $\ord{\tau}=2$. 
\end{example}

\begin{example}
Consider the the lazy CA $\tau : A^\mathbb{Z} \to A^\mathbb{Z}$ with minimal neighborhood $S = \{-1,0,1\}$, unique active transition $p = 001 \in A^S$ and writing symbol $a=1 \in A$ (this is ECA rule 206). Condition (1.) of Theorem \ref{upperbound} involves the set
\[ -S_0 + S_1 = \{0 , 1\} + \{1\} = \{1,2\}.  \] 
Since words of any length over $-S_1 = \{ -1\}$ only evaluate to negative integers, then Theorem \ref{upperbound} does not give us an upper bound for the order of $\tau$. It may be easily seen that $\tau$ has indeed infinite order, as the configuration $x = \dots 000100 \dots \in A^{\mathbb{Z}}$ never stabilizes. 
\end{example}

\begin{example}
Consider the the lazy CA $\tau : A^\mathbb{Z} \to A^\mathbb{Z}$ with minimal neighborhood $S = \{0,1, 3\}$, unique active transition $p = 01\_0 \in A^S$ and writing symbol $a=1 \in A$ (we use the symbol $\_$ in the notation of $p$ to emphasize that $S$ is not an interval of integers). Condition (1.) of Theorem \ref{upperbound} involves the set
\[ -S_0 + S_1 = \{0, -3\} + \{1\} = \{1,-2\}.  \]
The evaluation of the word of length $1$ over $-S_1 = \{-1\}$ is not in $-S_0 + S_1$, but the evaluation of the word of length $2$ over $-S_1 = \{-1\}$ is in $-S_0 + S_1$; hence, Theorem \ref{upperbound} shows that the order of $\tau$ is at most $3$. In Section 3 we will show that quasi-constant patterns such as $p = 01\_0$ always define lazy CA whose order achieve the upper bound given by Theorem \ref{upperbound}.  
\end{example}

As an easy consequence of Theorem \ref{upperbound} we obtain the following sufficient conditions on a lazy cellular automaton for being idempotent. 

\begin{corollary}\label{cor-idem}
    With the notation of Theorem \ref{upperbound}, assume at least one of the following holds:
    \begin{enumerate}
     \item $S_a = \emptyset$;
        \item $S_a^{-1}\subseteq S_b^{-1}S_a$, for some $b\in A\setminus\{a\}$;
        \item $S_a^{-1}\subseteq S_{b_1}^{-1}S_{b_2}$, for some $b_1, b_2\in A\setminus\{a\}$ such that $b_1\neq b_2$.
    \end{enumerate}
    Then, $\tau$ is idempotent.
\end{corollary}

\begin{example}
Consider the lazy CA $\tau : A^\mathbb{Z} \to A^\mathbb{Z}$ with minimal neighborhood $S = \{-1,0,1,2\}$, unique active transition $p = 1010 \in A^S$ and writing symbol $a=1 \in A$. Since $S_1 = \{-1,1\} = S_1^{-1}$, then condition (2.) of Corollary \ref{cor-idem} holds, so $\tau$ is idempotent. 
\end{example}

\begin{example}
A counterexample of the converse of Corollary \ref{cor-idem} was found by Mar\'ia G. Maga\~na-Ch\'avez in private communication. Consider the lazy CA $\tau : A^\mathbb{Z} \to A^\mathbb{Z}$ with minimal neighborhood $S = \{-3, -1, 0, 3, 4\}$, unique active transition $p = 11010 \in A^S$ and writing symbol $a=1 \in A$. Then $S_0 = \{0,4\}$ and $S_1 = \{-3,-1,3\}$, so none of the conditions of Corollary \ref{cor-idem} holds. However, $\tau$ is idempotent, as it may be checked that no pattern of $A^{S+S}$ evaluates to $p$ under the local rule of $\tau$ (see \cite[Remark 3.3]{Idem}). 
\end{example}


\section{Lazy CA with a quasi-constant active transition}

Recall that a pattern $p \in A^S$ is \emph{quasi-constant} if it is not constant and there exists an element $r \in S$, called the \emph{non-constant element of $p$}, such that $p \vert_{S \setminus \{r\}}$ is constant. 

\begin{theorem}\label{th-main}
Let $\tau : A^G \to A^G$ be a lazy cellular automaton with unique active transition $p \in A^S$ and writing symbol $a \in A \setminus \{ p(e) \}$. Assume that $p$ is quasi-constant with non-constant element $r \in S$.
\begin{enumerate}
\item If $a \neq p(s)$ for all $s \in S$, then $\ord{\tau}=2$. 

\item If $r \neq e$ and $a= p(r)$, then $\ord{\tau}$ is finite if and only if there exists $n \geq 2$ such that $r^n \in S$. Moreover, in this case, 
\[ \ord{\tau} = \min \{ n \geq 2 : r^n \in S \}.  \]

\item If $r = e$ and $a = p(s)$ for all $s \in S \setminus \{e\}$, then $\ord{\tau}$ is finite if and only if there exists $n \geq 2$ such that for all words $w\in (S^{-1} \setminus\{e\})^{n-1}$ there exists a subword $v\sqsubseteq w$ such that $\theta(v) \in S$. In such case, the order of $\tau$ is the minimum $n$ satisfying this property. 
\end{enumerate} 
\end{theorem} 

Before proving this theorem, we shall mention some of its consequences. The following corollary corresponds to Theorem 2 in \cite{Idem}.

\begin{corollary}
With the assumptions of Theorem \ref{th-main}, $\tau : A^G \to A^G$ is idempotent if and only if one of the following holds:
\begin{enumerate}
\item $a \neq p(s)$ for all $s \in S$. 

\item $r \neq e$ and $r^2 \in S$. 

\item $r = e$ and $S = S^{-1}$. 
\end{enumerate} 
\end{corollary}
\begin{proof}
Parts (1.) and (2.) are clear consequences of Theorem \ref{th-main}. Observe that part (3.) in Theorem \ref{th-main} with $n=2$ says that all $s \in S \setminus \{e\}$ satisfy $s^{-1} \in S \setminus \{e\}$, which is equivalent to $S=S^{-1}$. 
\end{proof} 

Recall that the \emph{order} of an element $g \in G$, denoted by $\ord{g}$, is the minimum integer $n \geq 1$ such that $g^n = e$, or infinity, in case that no such integer exists. Recall that $g^i \neq g^j$ holds for all $1 \leq i < j < \ord{g}$.

\begin{corollary}
Let $n \geq 2$ be an integer such that there is $g \in G$ with $\ord{g} > n$. Then, there exists a lazy CA $\tau : A^G \to A^G$ with $\ord{\tau} = n$. 
\end{corollary}
\begin{proof}
Let $S := \{e, g, g^{n} \}$ and define a quasi-constant pattern $p \in A^S$ by $p(e)=p(g^n) = 0$ and $p(g)=1$. Let $\tau : A^G \to A^G$ be the lazy CA with unique active transition $p$ and writing symbol $1 \in A$. The result follows by Theorem \ref{th-main} (2.), using the assumption $\ord{g} > n$.    
\end{proof}

\begin{example}
Consider the the lazy CA $\tau : A^\mathbb{Z} \to A^\mathbb{Z}$ with minimal neighborhood $S = \{0,1,n\}$, with $n \geq 2$, unique active transition $p = 01\_\dots\_n \in A^S$ and writing symbol $a=1 \in A$. By Theorem \ref{th-main} (2.), $\ord{\tau} = n$. 
\end{example}

\begin{example}
Consider the the lazy CA $\tau : A^\mathbb{Z} \to A^\mathbb{Z}$ with minimal neighborhood $S = \{-k,0,1,k\}$, with $k \geq 2$, unique active transition $p = 0\_\dots\_10\_\dots\_0 \in A^S$ and writing symbol $a=0 \in A$. We want to find the smallest $n$ such that every word $w \in (-S \setminus \{0\})^{n-1} = \{-k,-1,k\}^{n-1}$ has a subword $v\sqsubseteq w$ such that $\theta(v) \in S$. If $k$ or $-k$ appear in the word $w$, then we may always take $v=k$ or $v=-k$, respectively. Hence, assume that $w=(-1, \dots, -1)$. It is clear that we need the length of $w$ to be at least $k$ so we can find $v\sqsubseteq w$ such that $\theta(v) = -k \in S$. Therefore, $\ord{\tau}=k+1$. 
\end{example}

In order to prove Theorem \ref{th-main}, let $\tau : A^G \to A^G$ be a lazy cellular automaton with minimal local map $\mu : A^S \to A$, unique active transition $p \in A^S$ that is quasi-constant with non-constant element $r \in S$, and writing symbol $a := \mu(p) \in A \setminus \{p(e)\}$. Without loss of generality, we assume that $p \in A^S$ is defined as follows: for all $s \in S$, 
\[ p(s) := \begin{cases} 1 & \text{ if } s= r \\
0 & \text{ if } s \neq r.
\end{cases} \] 
Because of Corollary \ref{cor-idem}, we only need to consider the case when $a  = p(s)$ for some $s\in S$ (as otherwise, $S_a = \emptyset$). The proof of Theorem \ref{th-main} is divided in two cases: when the non-constant element $r$ is different from the group identity $e$, and when they are equal. The following sections deal with each of these cases. 

\subsection{Non-constant element $r \neq e$}

In this section, we assume that the non-constant element $r \in S$ of $p \in A^S$ is different from the group identity $e \in S$. Then, we must have $1 = p(r) = a = \mu(p)$, since we are assuming that $p(e) = 0$. 


\begin{proposition}\label{lower_bound_r!=e}
    For any $n\geq 2$, $\ord{\tau} > n$ if and only if $r^j\notin S$ for all $j \in [2, n]$.
\end{proposition}
\begin{proof}
    Assume that $\ord{\tau} > n$. By Theorem \ref{upperbound}, there exists a word $w=(s^{-1}_1, \ldots, s^{-1}_{n-1})\in (S_1^{-1})^{n-1}$ such that for all $i \leq j$, 
       \[  s_i^{-1} \dots s_j^{-1} \notin S_0^{-1}S_1.  \] 
By the construction of $p$ in this section, we have $S_0 = S\setminus\{r\}$ and $S_1 = \{r\}$, so the word $w$ has the form $w = (r^{-1},r^{-1}, \dots,  r^{-1})$. Hence, the above condition is equivalent to
\[   r^{-k}\notin (S^{-1}\setminus \{r\}) r, \quad \forall k \in [1, n-1], \quad \Longleftrightarrow \quad  r^{j}\notin S, \quad \forall j \in [2, n]. \]

Conversely, assume that $r^j\notin S$ for all $j \in [2, n]$. Define $x\in A^G$ as follows:
    \[x(g) := \begin{cases}
        p(g) & \textup{ if } g\in S \\
        0 & \textup{otherwise}
    \end{cases}.\]
    \begin{claim}
        For all $m\in [0,  n-1]$, $(g \cdot \tau^m(x))\vert_S = p$ if and only if $g = r^{-m}$. 
    \end{claim}
    \begin{proof}
        We prove this by strong induction on $m$. For $m=0$, it is clear by the construction of $x$. Now, assume the existence of $m\in [0, n-2]$ such that
        \[(g \cdot \tau^\ell(x))\vert_S = p \iff g = r^{-\ell}, \quad \forall \ell\in [0, m].\]
        We shall prove the case $m+1$. By the induction hypothesis and construction of $x$, 
        \begin{equation}\label{eq1}
            \tau^{m+1}(x)(g) = \mu( (g \cdot \tau^m(x))\vert_S) =  \begin{cases}
            1 & \text{ if } g \in\{r^{-\ell}\}_{\ell = -1}^m \\
           0 & \text{ otherwise } 
            \end{cases}. 
        \end{equation}
        
        We will show that $(g \cdot \tau^{m+1})(x) \vert_S = p$ if and only if $g = r^{-(m+1)}$. First we show that $(r^{-(m+1)} \cdot \tau^{m+1})(x) \vert_S = p$. For any $s\in S$, (\ref{eq1}) implies that
        \begin{align}\label{eq4}
            \tau^{m+1}(x)(sr^{-(m+1)}) = 1 \iff sr^{-(m+1)}\in\{r^{-\ell}\}_{\ell = -1}^m. 
        \end{align}
        If $sr^{-(m+1)} = r^{-\ell}$, then $s = r^{m+1-\ell}$. Note that $m+1-\ell\in [ 1, m+2] $, where $m+2\leq n$ by the choice of $m \in [0, n-2]$. Since $r^j\notin S$ for all $j \in [2, n]$, it follows that $s = r^{m+1-\ell}$ if and only if $m + 1 - \ell = 1$. As a consequence, $\ell = m$.  Then, by \eqref{eq1} and \eqref{eq4}
        \begin{align*}
            \tau^{m+1}(x)(sr^{-(m+1)}) = 1 &\iff s = r, \\
            \tau^{m+1}(x)(sr^{-(m+1)}) = 0 &\iff s \in S\setminus\{r\}.
        \end{align*}
        Therefore, $(r^{-(m+1)}\cdot\tau^{m+1}(x))\vert_S = p$. Conversely, suppose that $(g\cdot\tau^{m+1}(x))\vert_S = p$ for some $g\in G$. It follows by the construction of $p$ that
        \begin{align*}
            \tau^{m+1}(x)(g) &= p(e) = 0, \\
            \tau^{m+1}(x)(rg) &= p(r) = 1.
        \end{align*}
        As a consequence of \eqref{eq1}, $g\notin\{r^{-\ell}\}_{\ell = -1}^m$ and $rg\in\{r^{-\ell}\}_{\ell = -1}^m$, which implies that $g\in\{r^{-\ell}\}_{\ell = 0}^{m+1}$. Thus, $g = r^{-(m+1)}$ and the result follows.
    \end{proof}
    It follows by the previous claim and Lemma \ref{le-equality} that $\tau^{n-1}\neq \tau^n$. Therefore, $\ord{\tau} > n$ by Corollary \ref{order-lowerbound} and Proposition \ref{period}.
\end{proof}

Theorem \ref{th-main} (2.) follows as a direct consequence of Proposition \ref{lower_bound_r!=e}.


\subsection{Non-constant element $r=e$}

In this section, we assume that $r=e$, so the writing symbol of $\tau$ is $a = 0 = p(s)$, for all $s \in S \setminus \{e \}$. 


\begin{proposition}\label{lower_bound_r=e}
    For any $n\geq 2$, $\ord{\tau} > n$ if and only if there exists a word $(s^{-1}_1, \dots, s^{-1}_{n-1}) \in (S^{-1}\setminus\{e\})^{n-1}$ such that 
    \[ s_i^{-1} \dots s_j^{-1} \not\in S, \quad \forall  i \leq j.  \]
\end{proposition}
\begin{proof}
    Suppose that $\ord{\tau} > n$. By Theorem \ref{upperbound}, there exists a word $w=(s^{-1}_1, \dots, s^{-1}_{n-1})\in (S_0^{-1})^{n-1}$ such that for all $i \leq j$, 
       \[  s_{i}^{-1} \dots s_j^{-1} \notin S_1^{-1}S_0.  \] 
By the construction of $p$ in this section, we have that $S_0 = S \setminus \{e\} $ and $S_1 = \{e\}$, so $s_{i}^{-1} \dots s_j^{-1} \notin S \setminus \{e\}$ for all $i \leq j$. As $s^{-1}_i \in S^{-1}\setminus\{e\}$, for all $i$, it is not possible that $s_i^{-1} = e$. If $s_{i}^{-1} \dots s_j^{-1} = e$ for some $i < j$, then $s_{i+1}^{-1} \dots s_{j}^{-1} = s_i \in S \setminus \{e\}$ which is a contradiction. Therefore, $s_{i}^{-1} \dots s_j^{-1} \notin S$ for all $i \leq j$, and the direct implication follows. 
    
 Conversely, assume that there exists a word $(s^{-1}_1, \dots, s^{-1}_{n-1}) \in (S^{-1}\setminus\{e\})^{n-1}$ such that 
    \[ s_{i}^{-1} \dots s_j^{-1} \not\in S, \quad \forall  i \leq j.  \]
  Define $x\in A^G$ as follows: for all $g\in G$, 
    \[x(g) := \begin{cases}
        1 & \textup{ if } g\in\{ s_i\cdots s_1 s_0 \}_{i=0}^{n-1} \\
        0 & \textup{otherwise}
    \end{cases}\]
    where $s_0 = e$. 

      \begin{claim}\label{sup2_r=e}
        For all $i\in [0, n-1]$, $(g\cdot \tau^{i}(x))\vert_S = p$ if and only if $g = s_{n-(i + 1)}\cdots s_1s_0$.
    \end{claim}
    \begin{proof}
   First of all, observe that the hypothesis  $s_i^{-1} \dots s_j^{-1} \not\in S$, for all $ i \leq j$ implies that 
   \begin{equation}\label{sup1_r=e}
   \{s_{n-j}\cdots s_1s_0\}_{j = k}^{n}\cap \bigcup_{i = 1}^{k}\big((S\setminus\{e\})s_{n-i}\cdots s_1s_0\big) = \emptyset,  \quad \forall k \in [1, n]. 
   \end{equation}
   Indeed, if $s_{n-j}\cdots s_1s_0 \in (S\setminus\{e\})s_{n-i}\cdots s_1s_0$, then $s_{n-j+1}^{-1} \dots s_{n-i}^{-1} \in S \setminus \{e\}$, contradicting the hypothesis.  
   
 We prove Claim \ref{sup2_r=e} by strong induction on $i$. For the base case $i = 0$, fix $k = 1$ in (\ref{sup1_r=e}), so we see that $x(ss_{n-1}\cdots s_1) = 0$ for all $s\in S\setminus\{e\}$ and $x(s_{n-1}\cdots s_1) = 1$ by the construction of $x$. It follows that $(s_{n-1}\cdots s_1\cdot x)\vert_S = p$. Conversely, assume that $(g\cdot x)\vert_S = p$ for some $g\in G$, so $x(g) = 1$ and $x(sg) = 0$ for all $s\in S\setminus\{e\}$. If follows by the construction of $p$ that $g = s_{\ell}\cdots s_1s_0$ for some $\ell\in [0, n-1]$ and $s_j\cdots s_1s_0\notin (S\setminus\{e\})g$ for all $j\in [0, n-1]$. If $\ell<n-1$,
        \[s_{\ell+1} = (s_{\ell+1}\cdots s_1s_0)(s_0^{-1} s_{1}^{-1} \dots s_\ell^{-1})\notin S\setminus\{e\},\]
        which contradicts that $s_{\ell+1}\in S\setminus\{e\}$. Therefore, $g = s_{n-1}\cdots s_1s_0$. 
        
  For the induction hypothesis, assume the existence of $m \in [ 0, n-2]$ such that the following holds 
        \[(g\cdot \tau^{i}(x))\vert_S = p \iff g = s_{n-(i+1)}\cdots s_1s_0, \quad \forall i \in [0, m].\]
        We shall prove the case $m+1$ by showing that $(g\cdot \tau^{m+1}(x))\vert_S = p$ if and only if $g = s_{n-(m+2)}\cdots s_1s_0$. By the construction of $x$, the induction hypothesis, and Definition \ref{lazy},
        \[\tau^{m}(x)(g) = \begin{cases}
            1 & \textup{ if } g\in\{s_i\cdots s_1s_0\}_{i = 0}^{n-(m+1)} \\
            0 & \textup{otherwise}
        \end{cases}.\]
        Since $(g\cdot \tau^{m}(x))\vert_S = p$ if and only if $g = s_{n-(m+1)}\cdots s_1s_0$, it follows that $\tau^{m}(x)(g)\neq \tau^{m+1}(x)(g)$ if and only if $g = s_{n-(m+1)}\cdots s_1s_0$. Then
        \begin{align}\label{eq3}
            \tau^{m+1}(x)(g) = \begin{cases}
            1 & \textup{ if } g\in\{s_i\cdots s_1s_0\}_{i = 0}^{n-(m+2)} \\
            0 & \textup{otherwise}
        \end{cases}.
        \end{align}
        As a consequence of fixing $k = m+2$ in (\ref{sup1_r=e}),
        \[s_{j}\cdots s_1s_0\notin (S\setminus\{e\})s_{n-(m+2)}\cdots s_1s_0, \quad \forall j \in [ 0, n-(m+2)].\]
        Thus, $\tau^{m+1}(x)(ss_{n-(m+2)}\cdots s_1s_0) = 0$ for all $s\in S\setminus\{e\}$. Since $\tau^{m+1}(x)(s_{n-(m+2)}\cdots s_1s_0) = 1$, it follows that
        \[\big(s_{n-{(m+2)}}\cdots s_1s_0\cdot\tau^{m+1}(x)\big)\vert_S = p.\]
      Conversely, assume that $(g\cdot\tau^{m+1})(x)\vert_S = p$ for some $g\in G$. Observe that $\tau^{m+1}(x)(g) = 1$ and $\tau^{m+1}(x)(sg) = 0$ for all $s\in S\setminus\{e\}$. By (\ref{eq3}), we have  that $g = s_{\ell}\cdots s_1s_0$ for some $\ell\in [ 0, n - (m+2)]$ and $s_j\cdots s_1s_0\notin (S\setminus\{e\})g$ for all $j \in [ 0, n - (m+2)]$. If $\ell<n - (m+2)$, then
        \[s_{\ell+1} = (s_{\ell+1}\cdots s_1s_0)(s_0^{-1} s_1^{-1} \dots s_{\ell}^{-1}) \notin S\setminus\{e\},\]
        which contradicts that $s_{\ell+1}\in S\setminus\{e\}$. Therefore, $g = s_{n-(m+2)}\cdots s_1s_0$, and the claim follows.
    \end{proof}
    The previous claim and Corollary \ref{cor-power} imply that $\tau^{n-1}\neq \tau^{n}$. The result follows by Corollary \ref{order-lowerbound} and Proposition \ref{period}.
\end{proof}

Finally, Theorem \ref{th-main} (3.) follows as a direct consequence of Proposition \ref{lower_bound_r=e}.


\section{Open problems}

In this section, we propose two open problems related to the study of lazy CAs. 

In this paper, Corollary \ref{cor-idem} gives a sufficient condition for the idempotency of a lazy cellular automaton. In \cite{One-dim}, Conjecture 4.1 proposes that a lazy cellular automaton $\tau : A^G \to A^G$ with a unique active transition $p \in A^S$ is not idempotent if and only if $p$ satisfies a self-overlapping condition. The conjecture is true when $S \subseteq \mathbb{Z}$ is an interval, but it turns out that the direct implication fails even for general one-dimensional lazy cellular automata. Hence, we propose our first problem. 

\begin{problem}
Characterize the idempotency of a lazy cellular automaton $\tau : A^G \to A^G$ in terms of its unique active transition $p \in A^S$ and writing symbol $a \in A \setminus \{p(e)\}$.  
\end{problem}

The second problem we propose comes from the theory of monoids. Since the composition of two cellular automata over $A^G$ is again a cellular automaton over $A^G$, the set of all cellular automata over $A^G$, denoted by $\text{CA}(G,A)$ or $\text{End}(A^G)$, is a monoid equipped with the composition of functions. The group of all invertible cellular automata over $A^G$ is denoted by $\text{ICA}(G,A)$ or $\text{Aut}(A^G)$. For any subset $\mathcal{C}$ of cellular automata over $A^G$, denote by $\langle \mathcal{C} \rangle$ the submonoid of $\text{CA}(G,A)$ generated by $\mathcal{C}$.

It is well-known that the full transformation monoid $\text{Tran}(A)$ of a finite set $A$ is generated by the idempotents of defect $1$ (i.e., self-maps of $A$ whose image has size $\vert A \vert -1$) together with the group of invertible transformations $\text{Sym}(A)$ \cite{Howie}. As the minimal local maps of lazy cellular automata over $A^G$ with minimal neighborhood $S=\{e\}$ are precisely the idempotents of defect $1$ of $A$, the previous result inspired the following problem.  

\begin{problem}\label{prob2}
If $\mathcal{L}(G,A)$ is the set of all lazy cellular automata over $A^G$, prove or disprove the following:
\[ \text{CA}(G,A) = \langle \text{ICA}(G,A) \cup \mathcal{L}(G,A) \rangle.   \] 
If the above does not hold, what can we say of the submonoids $\langle \text{ICA}(G,A) \cup \mathcal{L}(G,A) \rangle$ and $\langle \mathcal{L}(G,A) \rangle$?
\end{problem}

In other words, Problem \ref{prob2} asks if every cellular automaton over $A^G$ may be written as a composition of lazy and invertible CAs. 

\begin{acknowledgments}
The first author was supported by SECIHTI \textit{Becas nacionales para estudios de posgrados.} We sincerely thank Nazim Fat\`es for suggesting the name \emph{lazy} for the class of cellular automata studied in this paper during the 31st International Workshop on Cellular Automata and Discrete Complex Systems AUTOMATA 2025, at the University of Lille, France. 

\end{acknowledgments}

\bibliographystyle{splncs04}

\end{document}